\documentclass{article}
\usepackage{amssymb}

\usepackage{amsmath}
\usepackage{graphicx,color}
\usepackage{epstopdf}
\usepackage{epsfig}
\usepackage{mathrsfs}

\newtheorem{theorem}{Theorem}

\newtheorem{definition}[theorem]{Definition}

\newtheorem{lemma}[theorem]{Lemma}

\newtheorem{proposition}[theorem]{Proposition}

\newenvironment{proof}[1][Proof]{\noindent\textbf{#1.} }{\ \rule{0.5em}{0.5em}}

\begin{document}

\title{The Series Product for Gaussian Quantum Input Processes }
\author{John E. Gough \\
Aberystwyth University, SY23 3BZ, Wales, United Kingdom \\
e-mail: jug@aber.ac.uk \\
[2ex] Matthew R. James \\
Australian National University, Canberra, ACT 0200, Australia \\
e-mail: Matthew.James@anu.edu.au}
\maketitle

\begin{abstract}
We present a theory for connecting quantum Markov components into a network
with quantum input processes in a Gaussian state (including thermal and
squeezed). One would expect on physical grounds that the connection rules
should be independent of the state of the input to the network. To compute
statistical properties, we use a version of Wicks' Theorem involving
fictitious vacuum fields (Fock space based representation of the fields) and
while this aids computation, and gives a rigorous formulation, the various
representations need not be unitarily equivalent. In particular, a naive
application of the connection rules would lead to the wrong answer. We
establish the correct interconnection rules, and show that while the quantum
stochastic differential equations of motion display explicitly the
covariances (thermal and squeezing parameters) of the Gaussian input fields
We introduce the Wick-Stratonovich form which leads to a way of writing
these equations that does not depend on these covariances and so corresponds
to the universal equations written in terms of formal quantum input
processes. We show that a wholly consistent theory of quantum open systems
in series can be developed in this way, and as required physically, is
universal and in particular representation-free.
\end{abstract}

\noindent \textbf{Keywords:} Gaussian Wick Theorem, Wick-Stratonovich Form,
Quantum Gaussian Feedback Networks.

\section{Introduction}

The quantum input-output theory has had an immense impact on quantum optics,
and in recent years has extended to opto-mechanical systems and beyond. The
prospect of routing the inputs through a network, or indeed using feedback
has lead to a burgeoning field of quantum feedback control \cite{I_12}-\cite
{JAM_14}. The development of a systems engineering approach to quantum
technology has benefited from having a systematic framework in which
traditional open quantum systems models can be combined according to
physical connection architectures.

The initial work on how to cascade two quantum input-output systems can be
traced back to Gardiner \cite{Gardiner_cascade} and Carmichael \cite
{Carmichael_cascade}. More generally, the authors have introduced the theory
of \textit{Quantum Feedback Networks} (QFN) which generalizes this to
include cascading, feedback, beam-splitting and general scattering of
inputs, etc., \cite{GJ-QFN}, \cite{GJ-Series}. One of the basic constructs
is the series product which gives the instantaneous feedforward limit of two
components connected in series via quantum input processes: in fact, the
systems need not necessarily be distinct and the series product generalizes
cascading by allowing for feedback. The original work was done for input
processes where the input fields where in the Fock vacuum field state. A
generalization to squeezed fields and squeezing components has been given 
\cite{GJN_squeeze}, however this was restricted to the case of linear
coupling and dynamics: there it was shown that the resulting transform
analysis could be applied in a completely consistent manner. More recent
work has shown that non-classical states for the input fields, such as
shaped single-photon or multi-photon states, or cat states of coherent
fields, may in principle be generated from signal models \cite{GJNC}, \cite
{GZ_filter} - that is, where a field in the Fock vacuum state was passed
through an ancillary dynamical system (the signal generator) which is then
to be cascaded to the desired system. Quantum feedback network (QFN) theory
concerns the interconnection of open quantum systems. The interconnections
are mediated by quantum fields in the input-output theory, \cite
{GarZol00,GJ-QFN,GJ-Series}. The idea is that an output from one node is fed
back in as input to another (not necessarily distinct) node, the simplest
case being the cascade connection (e.g., light exiting one cavity being
directed into another). The components are specified by Markovian models
determined by SLH parameters which describe the self-energy of the system
and how the system interacts with the fields (via idealized Jaynes-Cummings
type interactions and scattering).

Here we turn to the problem of the general class of Gaussian states for
quantum fields. This includes thermal fields, and of course squeezed fields.
In principle, these may be approximated as the output of a degenerate
parametric amplifier (DPA) driven by vacuum input, see \cite{GarZol00}. In a
sense, we have that a singular DPA may serve is the appropriate signal
generator to modify a vacuum field into a squeezed field before passing into
a given network. We will exploit this in the paper, however, we will have to
pay attention to the operator ordering problem when inserting these
approximations into quantum dynamical equations of motion and input-output
relations.

The programme turns out to be rather more involved than one might expect at
first glance. It is always possible to represent a collection of $d$
Gaussian fields using $2d$ vacuum fields (a Bogoliubov transformation!) and
one might hope that the corresponding connection rules applied to the
representation in terms of vacuum fields would agree with the intuitive
rules one would desire. This turns out not to be the case, and the various
feedback constraints cannot be naively applied to the representing fields:
the reason is that the representations are a linear combination of creation
and annihilation operators for the representing vacuum fields, and we have
broken the Wick ordered form of the original equations.

If applied naively, the series product would predict a contribution to the
global network model that depended on the covariance parameters of the
state. From the physical point of view, this ought to be spurious. In
comparison with classical analog linear electronics, we see that the
components (e.g. resistors, capacitors, inductors) are described by
impedances. When components are interconnected to form a network, the
network may be described by an equivalent impedance, derived through an
application of Kirchhoff laws. Impedances do not depend on the applied
currents or voltages, and are therefore intrinsic to the device or network.
Similarly the rules for connecting a quantum feedback network should be
intrinsic, and not depend on the state of the noise fields.

\section{Background and Problem statement}

Let us begin in the concrete setting of the quantum stochastic calculus of
Hudson and Parthasarathy \cite{HP} with a fixed initial space $\mathfrak{h}%
_{0}$ and a noise space that is the (Bose) Fock space over $\mathbb{C}^{d}$%
-valued $L^{2}$-functions on the time interval $[0,\infty )$. In the
language of Hudson and Parthasarathy, we have a multiplicity space of
dimension $d$ and we select an orthonormal basis which determines $d$
channels. We denote by $A_{k}\left( t\right) $, $A_{k}\left( t\right) ^{\ast
}$, and $\Lambda _{jk}\left( t\right) $ the processes of annihilation,
creation (for channel $j$) and scattering (from channel $k$ to channel $j$).
In the following, we shall introduce an Einstein summation convention for
repeated channel indices. We will deal with the class of quantum stochastic
integrals processes satisfying the appropriate conditions of local
integrability, square-integrability \cite{HP} without explicit reference.
We have for instance the QSDE 
\begin{equation}
dX\left( t\right) =x_{jk}\left( t\right) d\Lambda _{jk}\left( t\right)
+x_{j0}\left( t\right) dA_{j}\left( t\right) ^{\ast }+x_{0k}dA_{k}\left(
t\right) +x_{00}\left( t\right) dt
\end{equation}
where the coefficients are adapted and the increments are (quantum) It\={o}.
We have the quantum It\={o} product formula 
\begin{equation}
d\left( X\left( t\right) Y\left( t\right) \right) =dX\left( t\right)
\,Y\left( t\right) +X\left( t\right) \,dY\left( t\right) +dX\left( t\right)
\,dY\left( t\right)
\end{equation}
where the It\={o} correction comes from the quantum It\={o} table \cite{HP} 
\begin{eqnarray}
d\Lambda _{jk}\left( t\right) \,d\Lambda _{lm}\left( t\right) &=&\delta
_{kl}d\Lambda _{jm}\left( t\right) ,\quad d\Lambda _{jk}\left( t\right)
\,dA_{l}\left( t\right) ^{\ast }=\delta _{kl}dA_{j}\left( t\right) ^{\ast }, \nonumber
\\
dA_{k}\left( t\right) \,d\Lambda _{lm}\left( t\right) &=&\delta
_{kl}dA_{m}\left( t\right) ,\quad dA_{j}\left( t\right) \,dA_{k}=\delta
_{jk}dt,
\end{eqnarray}
with all other products of the fundamental increments vanishing.

\begin{definition}{Definition}
The Stratonovich integral is defined algebraically via 
\begin{eqnarray}
X\left( t\right) \circ dY\left( t\right) &=&X\left( t\right) dY\left(
t\right) +\frac{1}{2}dX\left( t\right) \,dY\left( t\right) \\
dX\left( t\right) \circ Y\left( t\right) &=&dX\left( t\right) \,Y\left(
t\right) +\frac{1}{2}dX\left( t\right) \,dY\left( t\right) .
\end{eqnarray}
\end{definition}

This turns out to be equivalent to a mid-point rule 
\cite{Chebotarev}.

If we consider the QSDE $dU\left( t\right) =-idE\left( t\right) \circ
U\left( t\right) $, with $U\left( 0\right) $ the identity and $E\left(
t\right) =E_{jk}\Lambda _{jk}\left( t\right) +E_{j0}B_{j}\left( t\right)
^{\ast }+E_{0k}B_{k}\left( t\right) +E_{00}$ a self-adjoint quantum
stochastic integral process, then we may convert to the It\={o} form to get 
\begin{equation}
dU\left( t\right) =\bigg\{\left( S_{jk}-\delta _{jk}\right) d\Lambda
_{jk}\left( t\right) +L_{j}dA_{j}^{\ast }\left( t\right) -L_{j}^{\ast
}S_{jk}dA_{k}\left( t\right) +Kdt\bigg\}\,U\left( t\right) ,
\end{equation}
where (setting $E_{\ell \ell }$ to be the $d\times d$ matrix with entries $%
E_{jk\text{)}}$%
\begin{equation}
S=\left[ 
\begin{array}{ccc}
S_{11} & \cdots & S_{1d} \\ 
\vdots & \ddots & \vdots \\ 
S_{d1} & \cdots & S_{dd}
\end{array}
\right] =\frac{I-\frac{i}{2}E_{\ell \ell }}{I+\frac{i}{2}E_{\ell \ell }}
\end{equation}
is called the matrix of scattering coefficients unitary (that is, $%
S_{jk}^{\ast }S_{jl}=\delta _{kl}=S_{lj}S_{kj}^{\ast }$), 
\begin{equation}
L=\left[ 
\begin{array}{c}
L_{1} \\ 
\vdots \\ 
L_{d}
\end{array}
\right] =\frac{i}{I+\frac{i}{2}E_{\ell \ell }}\left[ 
\begin{array}{c}
E_{10} \\ 
\vdots \\ 
E_{d0}
\end{array}
\right]
\end{equation}
which is the column vector of coupling operators, and 
\begin{equation}
K=-\frac{1}{2}L_{k}^{\ast }L_{k}-iH,
\label{eq:K_Fock}
\end{equation}
where $H$ is the Hamiltonian ($H^{\ast }=H=E_{00}+\frac{1}{2}E_{0j}\left[ 
\text{Im}\left\{ \frac{1}{I+\frac{i}{2}E_{\ell \ell }}\right\} \right]
_{jk}E_{k0}$). For simplicity we will assume that the terms $S_{jk},L_{j}$
and $H$ are bounded operators on the system Hilbert space $\mathfrak{h}_{0}$.

We generally refer to the triple $\mathbf{G}\sim \left( S,L,H\right) $ as
the Hudson-Parthasarathy parameters, or informally the ``SLH'' parameters
specifying the model. The unitary process they generate may be denoted as $%
U^{\mathbf{G}}\left( t\right) $ if we wish to emphasize the dependence on
these parameters.

For $X$ an operator of the initial space, we introduce $j_{t}\left( X\right)
=U\left( t\right) ^{\ast }X\,U\left( t\right) $ and from the quantum It\={o}
rule obtain the Heisenberg-Langevin equation 
\begin{equation}
dj_{t}(X)=j_{t}\left( \mathcal{L}_{jk}X\right) \,d\Lambda _{jk}+j_{t}(%
\mathcal{L}_{j0}X)\,dA_{j}^{\ast }+j_{t}(\mathcal{L}_{0k}X)\,dA_{k}+j_{t}(%
\mathcal{L}_{00}X)dt
\end{equation}
where 
\begin{equation}
\mathcal{L}_{jk}X=S_{lj}^{\ast }XS_{lk}-\delta _{jk}X,\quad \mathcal{L}%
_{j0}X=S_{lj}^{\ast }\left[ X,L_{l}\right] ,\quad \mathcal{L}_{0k}X=\left[
L_{l}^{\ast },X\right] S_{lk}
\end{equation}
and the Lindblad generator $\mathcal{L}_{00} \equiv \mathcal{L}$ is 
\begin{equation}
\mathcal{L} X=\frac{1}{2}L_{k}^{\ast }\left[ X,L_{k}\right] +\frac{1}{2}%
\left[ L_{k}^{\ast },X\right] L_{k}-i\left[ X,H\right] .
\label{eq:Linblad_Fock}
\end{equation}
The maps $\mathcal{L}_{\alpha \beta }$ are known as the \textit{Evans-Hudson
super-operators}. We shall occasionally write $j_{t}^{\mathbf{G}}\left(
X\right) $ for the dynamical flow of $X$ when we wish to emphasis the
dependence on the SLH parameters $\mathbf{G}$.

Let us now write the input processes as $A_{\mathrm{in},j}\left( t\right)
=A_{j}\left( t\right) $ and introduce the output processes as $A_{\mathrm{out%
},j}\left( t\right) =U\left( t\right) ^{\ast }A_{\mathrm{in},j}\left(
t\right) U\left( t\right) $ then from the quantum It\={o} rule we see that 
\begin{equation}
dA_{\mathrm{out},j}\left( t\right) =j_{t}\left( S_{jk}\right) \,dA_{\mathrm{%
in},k}\left( t\right) +j_{t}\left( L_{l}\right) \,dt.
\end{equation}

\subsection{Thermal Fields}

Considering the single channel $\left( d=1\right) $ case for the moment, we
may introduce non-Fock quantum stochastic processes as follows \cite{HL}. For $n>0$,
we set 
\begin{equation}
B\left( t\right) =\sqrt{n+1}A_{+}\left( t\right) +\sqrt{n}A_{-}\left(
t\right) ^{\ast },\quad \tilde{B}\left( t\right) =\sqrt{n}A_{+}\left(
t\right) +\sqrt{n+1}A_{-}\left( t\right) ^{\ast }
\end{equation}
which are canonical fields on the Fock space with a pair of channels
labeled as $k=\pm $. In fact, the map $\left( A_{+},A_{-}\right) \mapsto
\left( B,\tilde{B}\right) $ is a Bogoliubov transformation with inverse 
\begin{equation}
\left[ 
\begin{array}{c}
A_{+} \\ 
A_{-}
\end{array}
\right] =\left[ 
\begin{array}{cc}
\sqrt{\left( n+1\right) } & -\sqrt{n} \\ 
-\sqrt{n} & \sqrt{\left( n+1\right) }
\end{array}
\right] \left[ 
\begin{array}{c}
B \\ 
\tilde{B}
\end{array}
\right] .
\end{equation}

This is of course based on an Araki-Woods representation of the fiels \cite{AW63}. As is well known, these transformation cannot be implemented unitarily.
However, from a quantum optics point of view, devices transforming or even
squeezing fields in this manner are frequently considered, and it is useful
to imagine a hypothetical device - a Bogoliubov box - performing such a
canonical transformation on our idealized fields.

Ignoring the second process $\tilde{B}$, we obtain the non-Fock quantum
It\={o} table 
\begin{equation}
dB\left( t\right) dB\left( t\right) ^{\ast }=\left( n+1\right) dt,\quad
dB\left( t\right) ^{\ast }dB\left( t\right) =ndt.
\end{equation}
It problematic (read impossible) to incorporate a scattering process $%
\Lambda $ into this table. We refer to $B$ as non-Fock quantum noise.

We need to drop the scattering term from the unitary evolution equation,
i.e. set $S\equiv I$, and with 
\begin{equation}
L=\left[ 
\begin{array}{c}
L_{+} \\ 
L_{-}
\end{array}
\right] =\left[ 
\begin{array}{c}
\sqrt{n+1}L \\ 
-\sqrt{n}L^{\ast }
\end{array}
\right] 
\end{equation}
we have 
\begin{eqnarray}
dU\left( t\right)  &=&\bigg\{LdB\left( t\right) ^{\ast }-L^{\ast }dB\left(
t\right) +K^{\text{th}}dt\bigg\}\,U\left( t\right) \nonumber  \\
&=&\bigg\{L_{j}dA_{j}^{\ast }\left( t\right) -L_{j}^{\ast }dA_{j}\left(
t\right) +Kdt\bigg\}\,U\left( t\right) ,
\end{eqnarray}
where 
\begin{equation}
K^{\text{th}}=-\frac{1}{2}L_{+}^{\ast }L_{+}-\frac{1}{2}L_{-}^{\ast
}L_{-}-iH=-\frac{n+1}{2}L^{\ast }L-\frac{n}{2}LL^{\ast }-iH.
\end{equation}
For the flow, we need that the Hudson-Evans super-operator associated with
the scattering terms are trivial. This is the case when the entries of the
scattering matrix $S$ are (e.g. scalars) commuting with operators of the
initial space, but we can get away without assuming that $\left[ 
\begin{array}{cc}
S_{++} & S_{+-} \\ 
S_{-+} & S_{--}
\end{array}
\right] $ is the identity. By inspection we find that flow equation will
take the form 
\begin{equation}
dj_{t}\left( X\right) =j_{t}\left( \left[ X,L\right] \right) S^{\ast
}dB\left( t\right) ^{\ast }+j_{t}\left( \left[ L,X\right] \right) SdB\left(
t\right) +j_{t}\left( \mathcal{L}^{\text{th}}X\right) dt
\end{equation}
if and only if we take $\left[ 
\begin{array}{cc}
S_{++} & S_{+-} \\ 
S_{-+} & S_{--}
\end{array}
\right] \equiv \left[ 
\begin{array}{cc}
S & 0 \\ 
0 & S^{\ast }
\end{array}
\right] $ - otherwise we obtain the other noise $\tilde{B}$ - and in which
case the Lindbladian is 
\begin{eqnarray}
\mathcal{L}^{\text{th}}X &=&\frac{1}{2}\left[ L_{+}^{\ast },X\right] L_{+}+%
\frac{1}{2}L_{+}^{\ast }\left[ X,L_{+}\right] +\frac{1}{2}\left[ L_{-}^{\ast
},X\right] L_{-}+\frac{1}{2}L_{-}^{\ast }\left[ X,L_{-}\right] -i[X,H] \nonumber \\
&=&\frac{n+1}{2}\left\{ \left[ L^{\ast },X\right] L+L^{\ast }\left[ X,L%
\right] \right\} +\frac{n}{2}\left\{ \left[ L^{\ast },X\right] L+L^{\ast }%
\left[ X,L\right] \right\} -i[X,H].
\end{eqnarray}

\subsection{The Series Product - Vacuum Inputs}

In \cite{GJ-Series} the authors introduce a rule for combining SLH models in
series. For instance, we have the output
of the $\mathbf{G}_{\mathscr{A}}\sim \left( S_{\mathscr{A}},L_{\mathscr{A}%
},H_{\mathscr{A}}\right) $ fed instantaneously as input to $\mathbf{G}_{%
\mathscr{B}}\sim \left( S_{\mathscr{B}},L_{\mathscr{B}},H_{\mathscr{B}%
}\right) $ and it is shown that this is equivalent to the model generated by 
\begin{eqnarray}
\mathbf{G}_{\mathscr{B}}\vartriangleleft \mathbf{G}_{\mathscr{A}} &\sim
&\left( S_{\mathscr{A}},L_{\mathscr{A}},H_{\mathscr{A}}\right)
\vartriangleleft \left( S_{\mathscr{B}},L_{\mathscr{B}},H_{\mathscr{B}%
}\right)  \notag \\
&=&\bigg(S_{\mathscr{B}}S_{\mathscr{A}},L_{\mathscr{B}}+S_{\mathscr{B}}L_{%
\mathscr{A}},H_{\mathscr{A}}+H_{\mathscr{B}}+\mathrm{Im}\left\{ L_{%
\mathscr{B}}^{\ast }S_{\mathscr{B}}L_{\mathscr{A}}\right\} \bigg).
\label{eq:series_prod}
\end{eqnarray}
Here $\mathrm{Im}\left\{ C\right\} $ means $\frac{1}{2i}\left( C-C^{\ast
}\right) $.We note that every model may be written as a purely scattering
component and a non-scattering component in series, since we have the law $%
(S,L,H)=(I,L,H)\vartriangleleft (S,0,0)$.

We should remark that it is not necessary to view the two systems $%
\mathscr{A}$ and $\mathscr{B}$ as separate systems - specifically, in the
derivation of the series product\cite{GJ-Series} it is not assumed that the $%
\mathscr{A}$ and $\mathscr{B}$ operators need commute!

\subsection{Statement of the Problem}

If we wish to have a pair of systems $\mathscr{A}$ and $\mathscr{B}$ (both
accepting $d$ inputs) in series, then we obtain an equivalent Markov model
in the limit where the intervening connection is instantaneous. Let $L_{%
\mathscr{A}}$ be the column of the $d$ operators $L_{\mathscr{A},k}$, $%
k=1,\cdots , d$, and similar for system $\mathscr{B}$. The series product
says that the equivalent model has coupling $L_{\mathscr{A}}+L_{\mathscr{B}}$
and Hamiltonian 
\begin{eqnarray}
H_{\mathscr{A}}+H_{\mathscr{B}}+\mathrm{Im}\left\{ L_{%
\mathscr{B}}^{\ast }L_{\mathscr{A}}\right\} .
\end{eqnarray}

Suppose we were to apply the series product to two systems with the same
single thermal input $B$, and try and describe this as a series connection
using the two vacuum inputs $A_{+}$ and $A_{-}$. Naively applying the series
product to the construction in the $A_{\pm }$ format leads to the correct
rule $L_{\mathscr{A}}+L_{\mathscr{B}}$ for the coupling terms, but 
\begin{equation}
H_{\mathscr{A}}+H_{\mathscr{B}}+\mathrm{Im}\left\{ L_{\mathscr{B}}^{\ast }L_{%
\mathscr{A}}\right\} +n\mathrm{Im}\left[ L_{\mathscr{B}}^{\ast },L_{%
\mathscr{A}}\right] .
\end{equation}
We have picked up an $n$-dependent term. For pure cascading, the systems $%
\mathscr{A}$ and $\mathscr{B}$ are distinct and so $\left[ L_{\mathscr{B}%
}^{\ast },L_{\mathscr{A}}\right] =0$. However, the series product should
also apply to the situation where the systems share degrees of freedom. In
such cases the additional term is physically unreasonable as it depends on
the state of the noise.

It is not immediately obvious what is wrong with the construction. Going to
the double Fock vacuum representations and then using the vacuum version of
the series product would seem a reasonable thing to do. However, a fully
quantum description would involve the $\tilde{B}$ fields as well, and at a
schematic level this would involve one or more Bogoliubov boxes - something
conspicuously. We will give the correct procedure in this paper.

\section{Multi-Dimensional Gaussian Processes}

\subsection{Notation}

We will use the symbol $\triangleq $ to signify a defining equation. We will
denote the operations of complex conjugation, hermitean conjugation, and
more generally adjoint by *. For $X=[x_{ij}]$ an $n\times m$ array with
complex-valued entries, or more generally operator-valued entries, we write $%
X^{\ast }$ for the $m\times n$ array obtained by transposition of the array
and conjugation of the entries: that is the $ij$ entry is $x_{ji}^{\ast }$.
The transpose alone will be denotes as $X^{\top }$, that is the $m\times n$
array with $ij$ entry $x_{ji}$. We will also use the notation $%
X^{\#}=(X^{\top })^{\ast }$ which is the $n\times m$ array with $ij$ entry $%
x_{ij}^{\ast }$.

\subsection{Finite Dimensional Gaussian States}

\label{sec:FD_Gaussian}Let $a_{1},\cdots ,a_{d}$ be the annihilation
operators for $d$ independent oscillators. We consider a mean zero Gaussian
state with second moments 

\begin{equation}
n_{ij}=\langle a_{i}^{\ast }a_{j}\rangle ,\quad m_{ij}=\langle
a_{i}a_{j}\rangle ,
\end{equation}
which we assemble into a hermitean $d\times d$ matrix, $N$, with entries $%
n_{ji}^{\ast }=n_{ij}$, and a symmetric matrix $M$ is the $d\times d$ matrix
with entries $m_{ij}=m_{ji}$. The \emph{covariance matrix} is 
\begin{equation}
F=\left[ 
\begin{array}{cc}
I+N^{\top } & M \\ 
M^{\ast } & N
\end{array}
\right] .
\end{equation}
In order to yield mathematically correct variances, we must have both $F$
and $N$ positive. The vacuum state is characterized by having $N=M=0$, that
is 
\begin{equation}
F_{\mathrm{vac}}\equiv \left[ 
\begin{array}{cc}
I & 0 \\ 
0 & 0
\end{array}
\right] .
\label{eq:cov}
\end{equation}

The covariance matrix $F$ defined by (\ref{eq:cov}) must be positive
semi-definite, as will be the matrices $N$ and $I+N^{\top }$. We must also
have ran$\left( M\right) \subseteq $ran$\left( I+N^{\top }\right) $ and $%
MN^{-}M^{\ast }\leq I+N$, where $N^{-}$ is the Moore-Penrose inverse of $N$.

A linear transformation of the form 
\begin{eqnarray}
\tilde{a}=Ua+Va^{\#},
\end{eqnarray}
that is $\tilde{a}_{j}=\sum_{k}\left( U_{jk}a_{k}+V_{jk}a_{k}^{\ast }\right) 
$, is called a \textit{Bogoliubov transformation} if we have again the
canonical commutation relations for the primed operators.

The transformation $\tilde{a}=Ua+Va^{\#}$ is Bogoliubov if and only if the
following identities hold $UU^{\ast }=I+VV^{\ast },\quad UV^{\top }=VU^{\top
}.$

This is easily established by inspection, as are the following.

\begin{lemma}{Lemma}
Let $\tilde{a}=Ua+Va^{\#}$ be a Bogoliubov transformation, then the
covariance matrix for $\tilde{a}$ is 
\begin{eqnarray}
\tilde{F}=WFW^{\dag }
\end{eqnarray}
where $W=\Delta \left( U,V\right) $. In particular, the new matrices are 
\begin{eqnarray}
N^{\prime } &=&V^{\#}V^{\top }+V^{\#}N^{\top }V^{\top }+U^{\#}M^{\ast
}V^{\top }  +V^{\#}MU^{\top }+U^{\#}NU^{\top }, \nonumber\\
M^{\prime } &=&UV^{\top }+UN^{\top }V^{\top }+VM^{\ast }V^{\top }  +UM^{\ast }U^{\top }+VNU^{\top }.
\end{eqnarray}
\end{lemma}

\begin{lemma}{Lemma}
\label{Prop:W_vac}
 Given $a_{\mathrm{vac}}$ with the choice of the vacuum state, the
Bogoliubov transformation $a=Ua_{\mathrm{vac}}+Va_{\mathrm{vac}}^{\#}$ leads to
operators with the covariance matrix 
\begin{eqnarray}
F=WF_{\mathrm{vac}}W^{\ast }=\left[ 
\begin{array}{cc}
I+N^{\top } & M \\ 
M^{\ast } & N
\end{array}
\right]
\end{eqnarray}
where $W=\Delta \left( U,V\right) $ and 
\begin{eqnarray}
N=V^{\#}V^{\top },\quad M=UV^{\top }.
\end{eqnarray}
\end{lemma} 

We note that the determinant of the covariance matrix is preserved under
Bogoliubov transformations. In particular, if we have $F=WF_{\mathrm{vac}%
}W^{\ast }$, as in the last Proposition, then $F$ must also be singular.
This means that if we wish to obtain a given covariance matrix $F$ for $d$
modes by a Bogoliubov transformation of vacuum modes, we will typically need
a larger number $D$ of these modes with $F$ being a sub-block of a
transformed matrix $WF_{\mathrm{vac}}W^{\ast }$. The example in the Theorem
shows that in order to obtain the $d=1$ covariance 
\begin{eqnarray}
F=\left[ 
\begin{array}{cc}
1+n & 0 \\ 
0 & n
\end{array}
\right]
\end{eqnarray}
we need a Bogoliubov transformation of $D=2$ modes. We remark that we may
obtain the covariance 
\begin{eqnarray}
F=\left[ 
\begin{array}{cc}
1+n & m \\ 
m^{\ast } & n
\end{array}
\right] ,
\end{eqnarray}
with the constraint $\left| m\right| ^{2}\leq n\left( n+1\right) $ ensuring
positivity, from 2 vacuum modes via \cite{HHKKR02,G_QWN_ME} 
\begin{eqnarray}
\tilde{a}=\sqrt{n+1-\frac{1}{n}\left| m\right| ^{2}}a_{1}+\sqrt{n}%
a_{2}^{\ast }+\frac{m}{\sqrt{n}}a_{2}.  \label{eq:bog_m}
\end{eqnarray}
The maximal case $\left| m\right| ^{2}=n\left( n+1\right) $ may be obtained
from a \textit{single} mode $a_{1}$ via $a=\sqrt{n+1}a_{1}+e^{i\theta }\sqrt{%
n}a_{1}^{\ast }$ where $m\equiv \sqrt{n\left( n+1\right) }e^{i\theta }$.

\subsection{Quantum Ito Calculus: Gaussian Noise}

One would like to extend this to non-vacuum inputs, in particular, those
with general flat power Gaussian states for the noise. (We restrict to a
single noise channel for transparency but the generalization is
straightforward enough.) It is possible to construct noises having the
following quantum It\={o} table 
\begin{eqnarray}
dB_{i}dB_{j}^{\ast } &=&\left( n_{ji}+\delta _{ij}\right) dt,\quad
dB_{i}^{\ast }dB_{j}=n_{ij}dt,  \notag \\
dB_{i}dB_{j} &=&m_{ij}dt,\quad dB_{i}^{\ast }dB_{j}^{\ast }=m_{ji}^{\ast }dt,
\label{eq:table_non_Fock}
\end{eqnarray}
where $N=[n_{ij}]$ and $M=[m_{ij}]$ have the same properties and constraints
as introduced above.

In reality, we are assuming that the fields $B_{j}\left( t\right) $
correspond to a representation on a double Fock space, say, 
\begin{equation}
B(t)=U\left[ 
\begin{array}{c}
A_{+}\left( t\right) \otimes I \\ 
I\otimes A_{-}(t)
\end{array}
\right] +V\left[ 
\begin{array}{c}
A_{+}\left( t\right) ^{\#}\otimes I \\ 
I\otimes A_{-}(t)^{\#}
\end{array}
\right] 
\end{equation}
where $A_{k}\left( t\right) =\left[ 
\begin{array}{c}
A_{k,1}\left( t\right)  \\ 
\vdots  \\ 
A_{k.d}\left( t\right) 
\end{array}
\right] $ are copies of the Fock fields encountered above, and where $N=V^{\#}V,M=UV^{\top }$ as in Proposition 
\ref{Prop:W_vac}.

The underlying mathematical problem is that we are trying to implement a
canonical transformation that is not inner \cite{partha,Shale,DG}- specifically the various
representations for different pairs $\left( N,M\right) $ are not unitarily
equivalent.

Instead we must restrict to QSDE models in the general Gaussian case which
are driven by $B$ and $B^{\ast }$ only. We in fact find the class of QSDEs 
\begin{equation}
dU\left( t\right) =\left\{ L_{k}dB_{k}^{\ast }\left( t\right) -L_{k}^{\ast
}dB_{k}\left( t\right) +K^{\left( N,M\right) }dt\right\} \,U\left( t\right)
\label{eq:QSDE_non_Fock1}
\end{equation}
generating unitaries and we now require that 
\begin{equation}
K^{\left( N,M\right) }=-\frac{1}{2}(\delta _{ij}+n_{ji})L_{i}^{\ast }L_{j}-%
\frac{1}{2}n_{ij}L_{i}L_{j}^{\ast }+\frac{1}{2}m_{ij}L_{i}^{\ast
}L_{j}^{\ast }+\frac{1}{2}m_{ji}^{\ast }L_{i}L_{j}-iH,
\end{equation}
with $H$ again self-adjoint.

Let us denote the conditional expectation from the algebra of
operators on the system-tensor-Fock Hilbert space down to the system
operators (i.e., the partial trace over the Gaussian state) as $\mathbb{E}_{\left(
N,M\right) }\left[ \cdot |\mathrm{sys}\right] $. As the differentials $dB_{k}\left( t\right) $ and $%
dB_{k}\left( t\right) ^{\ast }$ are It\={o} (future pointing) their products
with adapted operators will have conditional expectation zero. Therefore 
\begin{equation}
\mathbb{E}_{\left( N,M\right) }\left[ dU_{t}|\mathrm{sys}\right] =K^{\left(
N,M\right) }\,\mathbb{E}_{\left( N,M\right) }\left[ U_{t}|\mathrm{sys}\right]
\,dt
\end{equation}
and we deduce that 
\begin{equation}
\mathbb{E}_{\left( N,M\right) }\left[ U_{t}|\mathrm{sys}\right]
=e^{tK^{\left( N,M\right) }}.
\end{equation}

The corresponding Heisenberg-Langevin equations are of the form 
\begin{equation}
dj_{t}(X)=j_{t}(\left[ X,L_{k}\right] )dB_{k}^{\ast }+j_{t}(\left[
L_{k}^{\ast },X\right] )dB_{k}+j_{t}(\mathcal{L}^{\left( N,M\right) }X)dt
\end{equation}
where the new Lindbladian is 
\begin{eqnarray}
\mathcal{L}^{\left( N,M\right) }X &=&\frac{1}{2}(\delta _{ij}+n_{ji})\big\{%
L_{i}^{\ast }\left[ X,L_{j}\right] +\left[ L_{i}^{\ast },X\right] L_{j}\big\}
\notag \\
&&+\frac{1}{2}n_{ij}\big\{L_{i}\left[ X,L_{j}^{\ast }\right] +\left[ L_{i},X%
\right] L_{j}^{\ast }\big\}  \notag \\
&&-\frac{1}{2}m_{ij}\big\{L_{i}^{\ast }\left[ X,L_{j}^{\ast }\right] +\left[
L_{i}^{\ast },X\right] L_{j}^{\ast }\big\}  \notag \\
&&-\frac{1}{2}m_{ji}^{\ast }\big\{L_{i}\left[ X,L_{j}\right] +\left[ L_{i},X%
\right] L_{j}\big\}-i\left[ X,H\right] .  \notag \\
&&  \label{eq:Lindblad_non_Fock}
\end{eqnarray}
Likewise, we find that 
\begin{equation}
\mathbb{E}_{\left( N,M\right) }\left[ j_{t}\left( X\right) |\mathrm{sys}%
\right] =e^{t\mathcal{L}^{\left( N,M\right) }}X.
\end{equation}

A little algebra allows us to relate these to the vacuum expressions: 
\begin{eqnarray}
K^{\left( N,M\right) } =K-\frac{1}{2}n_{ji}L_{i}^{\ast }L_{j}-\frac{1}{2}%
n_{ij}L_{i}L_{j}^{\ast } +\frac{1}{2}m_{ij}L_{i}^{\ast }L_{j}^{\ast }+\frac{1%
}{2}m_{ji}^{\ast }L_{i}L_{j},  \label{eq:K_form}
\end{eqnarray}
\begin{eqnarray}
\mathcal{L}^{\left( N,M\right) }X &=&\mathcal{L}X+\frac{1}{2}n_{ji}\big\{%
L_{i}^{\ast }\left[ X,L_{j}\right] +\left[ L_{i}^{\ast },X\right] L_{j}\big\}
\notag \\
&&+\frac{1}{2}n_{ij}\big\{L_{i}\left[ X,L_{j}^{\ast }\right] +\left[ L_{i},X%
\right] L_{j}^{\ast }\big\}  \notag \\
&&-\frac{1}{2}m_{ij}\big\{L_{i}^{\ast }\left[ X,L_{j}^{\ast }\right] +\left[
L_{i}^{\ast },X\right] L_{j}^{\ast }\big\}  \notag \\
&&-\frac{1}{2}m_{ji}^{\ast }\big\{L_{i}\left[ X,L_{j}\right] +\left[ L_{i},X%
\right] L_{j}\big\}  \notag \\
&\equiv &\mathcal{L}X+\frac{1}{2}n_{ji}\big\{\left[ L_{i}^{\ast },\left[
X,L_{j}\right] \right] +\left[ \left[ L_{i}^{\ast },X\right] ,L_{j}\right] %
\big\}  \notag \\
&&+\frac{1}{2}m_{ij}\left[ L_{j}^{\ast }\left[ L_{i}^{\ast },X\right] \right]
+\frac{1}{2}m_{ij}^{\ast }\left[ \left[ X,L_{i}\right] ,L_{j}\right] . 
\notag \\
&&  \label{eq:Lind_form}
\end{eqnarray}

\section{Representation-Free Form}

\label{Sec:Rep_Free} Returning to the problem stated in the Introduction, we
have that \textit{all} the $U_{t}^{\left( N,M\right) }$ arise from the \textit{same} physical dynamical evolution $U_{t}$,
and the dynamics show not depend on the state! The $U_{t}^{\left( N,M\right) }$ unfortunately belong to representations that are 
\textit{not} generally unitarily equivalent! There should be some sense in which the QSDEs for the
various $U_{t}^{\left( N,M\right) }$ should in some sense be equivalent.
These QSDEs will depend explicitly on the state parameters $\left(
N,M\right) $ of the input field, but what we would like to do is to show
that there is nevertheless a representation-free version of each of these
QSDEs in each fixed representation.

We now show that there is a way of presenting the unitary (\ref
{eq:QSDE_non_Fock1}) and Heisenberg (\ref{eq:Lindblad_non_Fock}) QSDEs so as
to be independent of the state parameters $(N,M)$.

\begin{theorem}{Theorem}
\textbf{(Representation-Free Form)} 
The non-Fock QSDEs (\ref{eq:QSDE_non_Fock1}) and (\ref{eq:Lindblad_non_Fock}%
) may be written in the equivalent Stratonovich forms 
\begin{eqnarray}
dU &=&dA_{k}^{\ast }\circ L_{k}U-L_{k}^{\ast }U\circ d A_{k}+KU\left( t\right) \circ dt,  \label{eq:Strat_QSDE} \\
dj_{t}(X) &=& d A_{k}^{\ast }\circ j_{t}(\left[ X,L_{k}\right] )+j_{t}(%
\left[ L_{k}^{\ast },X\right] )\circ d A_{k}  +j_{t}(\mathcal{L}X)\circ dt,  \label{eq:Strat_Heis}
\end{eqnarray}
respectively, where $K$ and $\mathcal{L}$ are the Fock representation
expressions (\ref{eq:K_Fock}) and (\ref{eq:Linblad_Fock}).
\end{theorem}

\begin{proof}
We first observe that 
\begin{equation}
dB_{k}^{\ast }\circ L_{k}U=dB_{k}^{\ast }L_{k}U+\frac{1}{2}dB_{k}^{\ast
}L_{k}dU
\end{equation}
and substituting the QSDE (\ref{eq:QSDE_non_Fock1}) for $dU$ and using the
quantum It\={o} table (\ref{eq:table_non_Fock}) gives 
\begin{equation}
dB_{k}^{\ast }\circ L_{k}U=L_{k}UdB_{k}^{\ast }+\frac{1}{2}L_{k}\left(
m_{kj}^{\ast }L_{j}-n_{kj}L_{j}^{\ast }\right) Udt,
\end{equation}
and similarly 
\begin{equation}
-L_{k}^{\ast }U\circ dB_{k}=-L_{k}^{\ast }UdB_{k}-\frac{1}{2}L_{k}^{\ast
}dUdB_{k}=-L_{k}^{\ast }UdB_{k}-\frac{1}{2}L_{k}^{\ast }\left(
n_{jk}L_{j}-m_{ki}L_{j}^{\ast }\right) dt.
\end{equation}
Combining these terms and using the identity (\ref{eq:K_form}) shows that (%
\ref{eq:Strat_QSDE}) is equivalent to (\ref{eq:QSDE_non_Fock1}).

For the Heisenberg equation, we first note that 
\begin{eqnarray}
dB_{k}^{\ast }\circ j_{t}(\left[ X,L_{k}\right] ) &=&dB_{k}^{\ast }j_{t}(%
\left[ X,L_{k}\right] )+\frac{1}{2}dB_{k}^{\ast }dj_{t}(\left[ X,L_{k}\right]
) \nonumber\\
&=&j_{t}(\left[ X,L_{k}\right] )dB_{k}^{\ast } \nonumber\\
&&+\frac{1}{2}dB_{k}^{\ast }%
\bigg\{j_{t}(\left[ \left[ X,L_{k}\right] ,L_{j}\right] )dB_{j}^{\ast
}+j_{t}\left( \left[ L_{j}^{\ast },\left[ X,L_{k}\right] \right] \right)
dB_{j}\bigg\} \nonumber\\
&=&j_{t}(\left[ X,L_{k}\right] )dB_{k}^{\ast }+j_{t}\big(\frac{1}{2}%
m_{kj}^{\ast }\big[\left[ X,L_{k}\right] ,L_{j}\big]+\frac{1}{2}n_{kj}\left[
L_{j}^{\ast },\left[ X,L_{k}\right] \right] \big)dt, \nonumber \\
\quad
\end{eqnarray}
and similarly 
\begin{equation}
j_{t}(\left[ L_{k}^{\ast },X\right] )\circ dB_{k}=j_{t}(\left[ L_{k}^{\ast
},X\right] )dB_{k}+j_{t}\bigg(\frac{1}{2}n_{jk}\big[\left[ L_{k}^{\ast },X%
\right] ,L_{j}\big]+\frac{1}{2}m_{jk}\left[ L_{j}^{\ast },\big[L_{k}^{\ast
},X\right] \big]\bigg)dt.
\end{equation}
Combining these terms and using the identity (\ref{eq:Lind_form}) shows that
(\ref{eq:Strat_Heis}) is equivalent to (\ref{eq:Lindblad_non_Fock}).
\end{proof}

Note that in both equations (\ref{eq:Strat_QSDE}) and (\ref{eq:Strat_Heis})
the Stratonovich differentials occur in Wick order relative to the integrand
terms. What is remarkable about these relations is that they are structurally the same as the Fock vacuum form of the QSDEs  with $S=I$. 
We say that the equations (\ref{eq:Strat_QSDE}) and (\ref{eq:Strat_Heis})
are \textit{representation-free} in the sense that they do not depend on the
parameters $N$ and $M$ determining the state of the noise.

\section{White Noise Description}

We now present a more formal, but insightful account of quantum stochastic
processes. Consider a collection of quantum noise input processes $%
\{b_{k}\left( t\right) :t\in \mathbb{R},k=1,\cdots ,d\}$ obeying the
commutation relations 
\begin{equation}
\left[ b_{j}\left( t\right) ,b_{k}^{\ast }\left( s\right) \right] =\delta
\left( t-s\right) ,\qquad \left[ b_{j}^{\ast }\left( t\right) ,b_{k}^{\ast
}\left( s\right) \right] =\left[ b_{j}\left( t\right) ,b_{k}\left( s\right) %
\right] =0.
\end{equation}
We wish to model the interaction of a quantum mechanical system driven by
these processes, and to this end introduce a unitary dynamics given by 
\begin{equation}
U\left( t\right) =\vec{\mathbf{T}}\exp \left\{ -i\int_{0}^{t}\Upsilon
_{s}ds\right\} 
\end{equation}
where (with an implied summation convention with range 1,$\cdots ,d$) 
\begin{equation}
-i\Upsilon _{t}=L_{k}\otimes b_{k}^{\ast }\left( t\right) -L_{k}^{\ast
}\otimes b_{k}\left( t\right) -iH\otimes I.
\end{equation}
Here $L_{k}$ and $H=H^{\ast }$ are system operators. The time ordering $\vec{%
\mathbf{T}}$ is understood in the usual sense of a Dyson series expansion.
From this we may arrive at 
\begin{equation}
\dot{U}\left( t\right) =L_{k}b_{k}^{\ast }\left( t\right) U\left( t\right)
-L_{k}^{\ast }b_{k}\left( t\right) U\left( t\right) -iHU\left( t\right) .
\label{eq:SCHROd}
\end{equation}

We claim that $U\left( t\right) $ should correspond to the evolution
operator for $\mathbf{G}\sim \left( S=I,L,H\right) $ without due reference
to a particular state for the noise. If we fix the state, say the vacuum,
then we use Wick ordering to compute the partial expectations with respect
to that state.

To see how to proceed, let us consider a general quantum stochastic integral 
$X\left( t\right) $ described by a formal equation
\begin{equation}
\dot{X}\left( t\right) =b_{j}\left( t\right) ^{\ast }x_{jk}\left( t\right)
b_{k}\left( t\right) +b_{j}\left( t\right) ^{\ast }x_{j0}\left( t\right)
+x_{0k}\left( t\right) b_{k}\left( t\right) +x_{00}\left( t\right) .
\label{eq:wn_qsi}
\end{equation}
where the terms $x_{\alpha \beta }\left( t\right) $ are ``adapted'' in the
formal sense that they do not depend on the noises $b_{k}\left( s\right) $
for $s>t$. As we are talking about the vacuum representation for the time
being, we can bootstrap from the vacuum $|\Omega \rangle $ to
construct the Fock space as the completion of the span of all vectors of the
type $\int f_{k\left( 1\right) }\left( t_{1}\right) b_{k\left( 1\right)
}\left( t_{1}\right) ^{\ast }\cdots f_{k\left( n\right) }\left( t_{n}\right)
b_{k\left( n\right) }\left( t_{m}\right) ^{\ast }|\Omega \rangle $, and
moreover we can build up the domain of exponential vectors. We quickly see
that (\ref{eq:wn_qsi}), with Wick ordered right hand side, corresponds to
the QSDE
\begin{equation}
dX\left( t\right) =x_{jk}\left( t\right) d\Lambda _{lk}\left( t\right)
+x_{j0}\left( t\right) dB_{j}\left( t\right) ^{\ast }+x_{0k}\left( t\right)
dB_{k}\left( t\right) +x_{00}\left( t\right) dt.
\end{equation}
Our issue however is how do we put to Wick order a given expression, for
instance, the right hand side of (\ref{eq:SCHROd}). 

\begin{proposition}{Proposition}
For the process $X\left( t\right) $ described by (\ref{eq:wn_qsi}), we have
\begin{eqnarray}
b_{k}\left( t\right) X\left( t\right)  &=&X\left( t\right) b_{k}\left(
t\right) +\frac{1}{2}x_{kl}\left( t\right) b_{l}\left( t\right) +\frac{1}{2}%
x_{k0}\left( t\right) ,  \notag \\
X\left( t\right) b_{k}\left( t\right) ^{\ast } &=&b_{k}\left( t\right)
^{\ast }X\left( t\right) +\frac{1}{2}b_{j}\left( t\right) ^{\ast
}x_{j0}\left( t\right) +\frac{1}{2}x_{0k}\left( t\right) .
\label{eq:wn_Strat}
\end{eqnarray}
\end{proposition}

\bigskip 

We may justify this as follows:
\begin{eqnarray}
\left[ b_{k}\left( t\right) ,X\left( t\right) \right]  &=&\int_{0}^{t}\left[
b_{k}\left( t\right) ,\dot{X}\left( s\right) \right] ds=\int_{0}^{t}\delta
\left( t-s\right) \left\{ x_{kl}\left( s\right) b_{l}\left( s\right)
+x_{k0}\left( s\right) \right\} \nonumber  \\
&=&\frac{1}{2}x_{kl}\left( t\right) b_{l}\left( t\right) +\frac{1}{2}%
x_{k0}\left( t\right) 
\end{eqnarray}
with the factor of $\frac{1}{2}$ coming from the half-contribution of the $%
\delta $-function. Evidently what the equations in (\ref{eq:wn_Strat})
correspond to is our definition of a Stratonovich differential - at least
for the Fock vacuum representation. While we can make a connection between (%
\ref{eq:wn_qsi}) and the rigorously defined Hudson-Parthasarathy processes,
it should be appreciated at the very least that (\ref{eq:wn_Strat}) is the
correct mnemonic for doing the Wick ordering - an attempt to convert into a
Dyson-type series expansion and Wick ordering under the iterated integral
signs to get a Maassen-Meyer kernel expansion shows this. At work here is an 
old principle that ``It\^{o}’s formula is the chain rule with Wick
ordering'' \cite{HS}. Let us now examine
(\ref{eq:SCHROd}) and put it to Wick ordered form. By a similar argument,
we have
\begin{equation}
\left[ b_{k}\left( t\right) ,U\left( t\right) \right] =\int_{0}^{t}\left[
b_{k}\left( t\right) ,\Upsilon \left( s\right) \right] U\left( s\right)
ds\equiv \frac{1}{2}L_{k}U\left( t\right) ,
\end{equation}
or $b_{k}\left( t\right) U\left( t\right) =U\left( t\right) b_{k}\left(
t\right) +\frac{1}{2}L_{k}U\left( t\right) $. By means of this we may place (%
\ref{eq:SCHROd}) into the Wick-ordered form 
\begin{equation}
U\left( t\right) =L_{k}b_{k}^{\ast }\left( t\right) U\left( t\right)
-L_{k}^{\ast }U\left( t\right) b_{k}\left( t\right) -(\frac{1}{2}L_{k}^{\ast
}L_{k}+iH)U\left( t\right) ,
\end{equation}
and picking up the correct vacuum damping (\ref{eq:K_Fock}), $K$, as a result.

Setting $X_{t}=U\left( t\right) (X\otimes I)U\left( t\right) $, the same
Wick ordering rule can be applied to the Heisenberg equations to obtain 
\begin{equation}
\dot{X}_{t}=\left\{ b_{k}^{\ast }\left( t\right) +\frac{1}{2}L_{k,t}^{\ast
}\right\} \left[ X,L_{k}\right] _{t}+\left[ L_{k}^{\ast },X\right]
_{t}\left\{ b_{k}\left( t\right) +\frac{1}{2}L_{k,t}\right\} +\frac{1}{i}%
U\left( t\right) \left[ X,H\right] U\left( t\right) .
\end{equation}
Here we use the notation $L_{k,t}=U\left( t\right) (L_{k}\otimes I)U\left(
t\right) $, etc.

We also remark that we may define the corresponding \textit{output fields}
by 
\begin{eqnarray}
b^{\mathrm{out}}_k (t) \triangleq U^\ast_T \, b(t) \, U_T,
\end{eqnarray}
where $T>t$. One may show that the input-output relations are 
\begin{eqnarray}
b^{\mathrm{out}}_k (t) \equiv b_k (t) + L_{k,t}.  \label{eq:i-o}
\end{eqnarray}

If, on the other hand, we want the state of the noise to be a mean-zero
Gaussian with correlations, say 
\begin{equation}
\left\langle b_{j}\left( t\right) ^{\ast }b_{k}\left( s\right) \right\rangle
=n_{jk}\,\delta \left( t-s\right) ,\quad \left\langle b_{j}\left( t\right)
b_{k}\left( s\right) \right\rangle =m_{jk}\,\delta \left( t-s\right) ,
\label{eq:cov_flat}
\end{equation}
then we represent the noise as 
\begin{equation}
b_{k}\left( t\right) =U_{jk}a_{+,k}\left( t\right) +V_{jk}a_{-,k}\left(
t\right) ^{\ast }  \label{eq:wn_bog}
\end{equation}
employing a suitable Bogoliubov transformation. Here we now have double the
number of quantum white noises $a_{+,k}$ and $a_{-,k}$ but these are
represented as Fock processes.

If we now substitute (\ref{eq:wn_bog}) into (\ref{eq:SCHROd}) we see
explicitly that the $a_{\pm ,k}$ are out Wick order, but this can be
rectified by the same sort of manipulation as above. Once the $a_{\pm
,k}\left( t\right) $ are Wick ordered, we have a equation which we can
interpret as the It\={o} non-Fock QSDE, and this leads to the correct
expressions $K^{\left( N,M\right) }$ and $\mathcal{L}^{\left( N,M\right) }$
in the unitary and flow equations respectively.

Given a Gaussian state $\left\langle \cdot \right\rangle $ on the noise, we
may introduce a conditional expectation according to $\mathbb{E}\left[ \cdot
|\mathrm{sys}\right] :A\otimes B\mapsto \left\langle B\right\rangle \,A$.
For instance, $\mathbb{E}\left[ U\left( t\right) |\mathrm{sys}\right] $ then
defines a contraction on the system Hilbert space and we have 
\begin{equation}
\mathbb{E}\left[ U\left( t\right) |\mathrm{sys}\right] =I_{\mathrm{sys}%
}+\sum_{n\geq 1}\left( -i\right) ^{n}\int_{\Delta _{n}\left( t\right) }%
\mathbb{E}\left[ \Upsilon _{s_{n}}\cdots \Upsilon _{s_{1}}|\mathrm{sys}%
\right] .
\end{equation}
Now the expression $\mathbb{E}\left[ \Upsilon _{s_{n}}\cdots \Upsilon
_{s_{1}}|\mathrm{sys}\right] $ will be a sum of products of the operators $%
L,-L^{\ast }$ and $H$ times a $n$-point function in the fields. Similarly,
we obtain a reduced Heisenberg equation. To compute these averages we need
to be able to calculate $n$-point functions of chronologically ordered
Gaussian fields - this is the realm of Wick's Theorem, so what we have
presented may be interpreted as a Gaussian Wick's Theorem \cite{Evans_Steer}. We of course recover the 
partial traces appearing in the previous section.

\section{Approximate Signal Generator for Thermal States}

In this section we show how to go from a general SLH model driven by
the output of a Degenerate Parametric Amplifier (DPA) to the limit where the same SLH model is driven by a
thermal white noise. We start with the single channel for simplicity.

\subsection{The Thermal White Noise as Idealization of the Output of a
Degenerate Parametric Amplifier}

We now show that in the strong coupling limit the output of a degenerate
parametric amplifier approximates a thermal white noise. the model consists
of a system of two cavities modes $c_{+}$ and $c_{-}$ coupled to input
processes $A_{+}\left( t\right) $ and $A_{-}\left( t\right) $ respectively.
Both inputs are taken to be in the vacuum state and the Schr\"{o}dinger
equation is 
\begin{equation}
\dot{U}_{t}=\sum_{i=+,-}L_{i}U\left( t\right) dA_{i}\left( t\right) ^{\ast
}-\sum_{i=+,-}L_{i}^{\ast }U\left( t\right) dA_{i}\left( t\right) -iH_{%
\mathrm{amp}}U_{t},
\end{equation}
with initial condition $U_{0}=I$ and 
\begin{equation}
L_{+}=\sqrt{2\kappa k}c_{+},\quad L_{-}=\sqrt{2\kappa k}c_{-}\text{ and }H_{%
\mathrm{amp}}=\frac{\varepsilon k}{i}\left( c_{+}c_{-}-c_{+}c_{-}\right) .
\end{equation}
Here $\varepsilon >\kappa $ and $k>0$ is a scaling parameter which we
eventually model to be large. It is more convenient to work with the white noises 
$a_{\pm }\left( t\right) $.

The model is linear and we obtain the input-output relations in the Laplace
domain to be \cite{GJN_squeeze} 
\begin{equation}
\left[ 
\begin{array}{c}
b\left[ s\right]  \\ 
\tilde{b}\left[ s\right] 
\end{array}
\right] =\Xi _{-}^{\left( k\right) }\left( s\right) \left[ 
\begin{array}{c}
a_{+}\left[ s\right]  \\ 
a_{-}\left[ s\right] 
\end{array}
\right] +\Xi _{+}^{\left( k\right) }\left( s\right) \left[ 
\begin{array}{c}
a_{+}\left[ s\right]  \\ 
a_{-}\left[ s\right] 
\end{array}
\right] 
\end{equation}
where $\Xi _{-}^{\left( k\right) }\left( s\right) =\left[ 
\begin{array}{cc}
u\left( s/k\right)  & 0 \\ 
0 & u\left( s/k\right) 
\end{array}
\right] ,\quad \Xi _{+}^{\left( k\right) }\left( s\right) =\left[ 
\begin{array}{cc}
0 & v\left( s/k\right)  \\ 
v\left( s/k\right)  & 0
\end{array}
\right] $ with the functions $u\left( s\right) =\frac{s^{2}-\kappa
^{2}-\varepsilon ^{2}}{s^{2}+2s\kappa +\kappa ^{2}-\varepsilon ^{2}},\quad
v\left( s\right) =\frac{2\kappa \varepsilon }{s^{2}+2\kappa +\kappa
^{2}-\varepsilon ^{2}}$.

In the limit $k\rightarrow \infty $ we find the static ($s$-independent)
coefficients 
\begin{equation}
\lim_{k\rightarrow \infty }\Xi _{-}^{\left( k\right) }\left( s\right) =\frac{%
\varepsilon ^{2}+\kappa ^{2}}{\varepsilon ^{2}-\kappa ^{2}}\left[ 
\begin{array}{cc}
1 & 0 \\ 
0 & 1
\end{array}
\right] ,\quad \lim_{k\rightarrow \infty }\Xi _{+}^{\left( k\right) }\left(
s\right) =\frac{2\varepsilon \kappa }{\varepsilon ^{2}-\kappa ^{2}}\left[ 
\begin{array}{cc}
0 & 1 \\ 
1 & 0
\end{array}
\right] .
\end{equation}
and returning to the time domain, the limit output fields are just a
Bogoliubov transform of the inputs 
\begin{equation}
b\left( t\right) =\sqrt{n+1}a_{+}\left( t\right) +\sqrt{n}a_{-}\left(
t\right) ,\quad \tilde{b}\left( t\right) =\sqrt{n}a_{+}\left( t\right) +%
\sqrt{n+1}a_{-}\left( t\right) ,
\end{equation}
Here the parameter $n$ corresponds is $n=\left( \frac{2\varepsilon \kappa }{%
\varepsilon ^{2}-\kappa ^{2}}\right) ^{2}.$

It is instructive to look closely at the finite $k$ equations. We have the
Heisenberg equations 
\begin{eqnarray}
\dot{c}_{+}\left( t\right)  &=&-k\kappa c_{+}\left( t\right) +k\varepsilon
c_{-}\left( t\right) -\sqrt{2\kappa k}a_{+}\left( t\right) , \nonumber \\
\dot{c}_{-}\left( t\right)  &=&-k\kappa c_{-}\left( t\right) +k\varepsilon
c_{+}\left( t\right) -\sqrt{2\kappa k}a_{-}\left( t\right) ,
\end{eqnarray}
and for $k$ large we may ignore the $\dot{c}_{+}\left( t\right) $ and $%
\dot{c}_{-}\left( t\right) $ terms leaving a pair of simultaneous equations
which we may solve to get 
\begin{eqnarray}
\sqrt{k}c_{+}\left( t\right) &\simeq &\frac{\sqrt{2\kappa }}{\varepsilon
^{2}-\kappa ^{2}}\left[ \kappa a_{+}\left( t\right) +\varepsilon a_{-}\left(
t\right) ^{\ast }\right] ,\nonumber \\
 \sqrt{k}c_{-}\left( t\right) &\simeq &\frac{%
\sqrt{2\kappa }}{\varepsilon ^{2}-\kappa ^{2}}\left[ \kappa a_{-}\left(
t\right) +\varepsilon a_{+}\left( t\right) ^{\ast }\right] .
\label{eq:approx_eq}
\end{eqnarray}
The output is then 
\begin{eqnarray}
b\left( t\right)  &=&a_{+}\left( t\right) +\sqrt{2\kappa k}c_{+}\left(
t\right) \simeq a_{+}\left( t\right) +\frac{2\kappa }{\varepsilon
^{2}-\kappa ^{2}}\left[ \kappa a_{+}\left( t\right) +\varepsilon a_{-}\left(
t\right) ^{\ast }\right]  \nonumber \\
&\equiv &\sqrt{n+1}a_{+}\left( t\right) +\sqrt{n}a_{-}\left( t\right) ,
\end{eqnarray}
and likewise 
\begin{eqnarray}
\tilde{b}\left( t\right)  &=&a_{-}\left( t\right) +\sqrt{2\kappa k}%
c_{-}\left( t\right) \simeq a_{-}\left( t\right) +\frac{2\kappa }{%
\varepsilon ^{2}-\kappa ^{2}}\left[ \kappa a_{-}\left( t\right) +\varepsilon
a_{+}\left( t\right) ^{\ast }\right] \nonumber \\
&\equiv &\sqrt{n}a_{+}\left( t\right) +\sqrt{n+1}a_{-}\left( t\right) .
\end{eqnarray}

It is relatively straightforward to find a multi-dimensional version of this
for a general Bogoliubov transformation 
\begin{equation}
\left[ 
\begin{array}{c}
b\left( t\right)  \\ 
\tilde{b}\left( t\right) 
\end{array}
\right] =U\left[ 
\begin{array}{c}
a_{+}\left( t\right)  \\ 
a_{-}\left( t\right) 
\end{array}
\right] +V\left[ 
\begin{array}{c}
a_{+}\left( t\right)  \\ 
a_{-}\left( t\right) 
\end{array}
\right] .
\end{equation}

\subsection{Cascade Approximation}

The DPA which is described by 
\begin{equation}
\mathbf{G}_{DPA}\sim \left( \left[ 
\begin{array}{cc}
1 & 0 \\ 
0 & 1
\end{array}
\right] ,\left[ 
\begin{array}{c}
\sqrt{2\kappa k}c_{+} \\ 
\sqrt{2\kappa k}c_{-}
\end{array}
\right] ,H_{\mathrm{amp}}\right) 
\end{equation}
driven by the (vacuum) input pair $\left[ 
\begin{array}{c}
a_{+}\left( t\right)  \\ 
a_{-}\left( t\right) 
\end{array}
\right] $. It is then put in series with 
\begin{equation}
\mathbf{G}\sim \left( S,L,H\right) \boxplus \left( 1,0,0\right) =\left( 
\left[ 
\begin{array}{cc}
S & 0 \\ 
0 & 1
\end{array}
\right] ,\left[ 
\begin{array}{c}
L \\ 
0
\end{array}
\right] ,H\right) 
\end{equation}
which means that the output $a_{+}\left( t\right) $ is fed in as input to
the system $\mathbf{G}\sim \left( S,L,H\right) $ and $a_{-}\left( t\right) $
is left to go away unhindered, $\mathbf{G}_{\mathrm{trivial}}\sim \left(
1,0,0\right) $. According to the series product rule, we get DPA and system
in series is described by, 
\begin{equation}
\mathbf{G}\vartriangleleft \mathbf{G}_{DPA}\sim \bigg(\left[ 
\begin{array}{cc}
S & 0 \\ 
0 & 1
\end{array}
\right] ,\left[ 
\begin{array}{c}
L+S\sqrt{2\kappa k}c_{+} \\ 
\sqrt{2\kappa k}c_{-}
\end{array}
\right] ,H+H_{\mathrm{amp}}+\frac{\sqrt{\kappa k}}{\sqrt{2}i}\left( L^{\ast
}Sc_{+}-c_{+}^{\ast }S^{\ast }L\right) \bigg).
\end{equation}

From this we obtain the Heisenberg equations 
\begin{eqnarray}
\dot{X}_{t} &=&a_{+}\left( t\right) ^{\ast }\left( S^{\ast }XS-X\right)
_{t}a_{+}\left( t\right) +a_{+}\left( t\right) ^{\ast }S_{t}^{\ast }\left[
X,L\right] _{t}+\left[ L^{\ast },X\right] _{t}S_{t}a_{+}\left( t\right) \nonumber \\
&&+\frac{1}{2}\left[ L^{\ast },X\right] _{t}\left( L+S\sqrt{2\kappa k}%
c_{+}\right) _{t}+\frac{1}{2}\left( L+S\sqrt{2\kappa k}c_{+}\right)
_{t}^{\ast }\left[ X,L\right] _{t} \nonumber \\
&&-i\left[ X,H+\frac{\sqrt{2\kappa k}}{2i}\left( L^{\ast }Sc_{+}-c_{+}^{\ast
}S^{\ast }L\right) \right] _{t}.
\end{eqnarray}

We now make the approximation $\sqrt{k}c_{+}\left( t\right) \simeq \frac{%
\sqrt{2\kappa }}{\varepsilon ^{2}-\kappa ^{2}}\left[ \kappa a_{+}\left(
t\right) +\varepsilon a_{-}\left( t\right) ^{\ast }\right] $ which leads to

\begin{eqnarray}
\dot{X}_{t} &\simeq &a_{+}\left( t\right) ^{\ast }\left( S^{\ast
}XS-X\right) _{t}a_{+}\left( t\right) +a_{+}\left( t\right) ^{\ast
}S_{t}^{\ast }\left[ X,L\right] _{t}+\left[ L^{\ast },X\right]
_{t}S_{t}a_{+}\left( t\right) +\mathcal{L}\left( X\right) _{t} \nonumber \\
&&+\left\{ \left[ L^{\ast },X\right] _{t}S_{t}+\frac{1}{2}L_{t}^{\ast }\left[
S,X\right] _{t}\right\} \left[ \left( \sqrt{n+1}-1\right) a_{+}\left(
t\right) +\sqrt{n}a_{-}\left( t\right) ^{\ast }\right]  \nonumber \\
&&+\left[ \left( \sqrt{n+1}-1\right) a_{+}\left( t\right) ^{\ast }+\sqrt{n}%
a_{-}\left( t\right) \right] \left\{ S_{t}^{\ast }\left[ X,L\right] _{t}+%
\frac{1}{2}\left[ X,S^{\ast }\right] _{t}L_{t}\right\} .
\end{eqnarray}
Here we have $n=\left( \frac{2\varepsilon \kappa }{\varepsilon
^{2}-\kappa ^{2}}\right) ^{2}$, as before.

We now make a key assumption: \textbf{the scattering term} $S$ \textbf{%
corresponds to a static element}. In this case $S\equiv e^{i\theta }$ for
some real $\theta $. The limit Heisenberg equation therefore simplifies to
\begin{eqnarray}
\dot{X}_{t} &=&a_{+}\left( t\right) ^{\ast }S^{\ast }\left[ X,L\right] _{t}+%
\left[ L^{\ast },X\right] _{t}Sa_{+}\left( t\right) +\mathcal{L}\left(
X\right) _{t} \nonumber \\
&&+\left[ L^{\ast },X\right] _{t}S\left[ \left( \sqrt{n+1}-1\right)
a_{+}\left( t\right) +\sqrt{n}a_{-}\left( t\right) ^{\ast }\right] \nonumber \\
&&+\left[
\left( \sqrt{n+1}-1\right) a_{+}\left( t\right) ^{\ast }+\sqrt{n}a_{-}\left(
t\right) \right] S^{\ast }\left[ X,L\right] _{t} \nonumber \\
&=&\sqrt{n+1}a_{+}\left( t\right) ^{\ast }S^{\ast }\left[ X,L\right] _{t}+%
\sqrt{n+1}\left[ L^{\ast },X\right] _{t}Sa_{+}\left( t\right)  \nonumber \\
&&+\sqrt{n}\left[ L^{\ast },X\right] _{t}Sa_{-}\left( t\right) ^{\ast }+%
\sqrt{n}a_{-}\left( t\right) S^{\ast }\left[ X,L\right] _{t}+\mathcal{L}%
\left( X\right) _{t}.
\end{eqnarray}

We are not quite finished as the operators $a_{-}\left( t\right) $ and $%
a_{-}\left( t\right) $ are out of Wick order. However, this is easily
remedied. For instance, we easily deduce that 
\begin{eqnarray}
\left[ Y_{t},a_{-}\left( t\right) ^{\ast }\right]  &=&\int_{0}^{t}\left[ 
\dot{Y}_{s},a_{-}\left( t\right) ^{\ast }\right] ds \nonumber \\
&=&\int_{0}^{t}\left[ \sqrt{n}a_{-}\left( s\right) S^{\ast }\left[ Y,L\right]
_{s},a_{-}\left( t\right) ^{\ast }\right] ds \nonumber \\
&=&\frac{1}{2}\sqrt{n}S^{\ast }\left[ Y,L\right] _{t}
\end{eqnarray}
so that we arrive at 
\begin{equation}
\left[ L^{\ast },X\right] _{t}Sa_{-}\left( t\right) ^{\ast }=a_{-}\left(
t\right) ^{\ast }\left[ L^{\ast },X\right] _{t}S+\frac{1}{2}\sqrt{n}\left[ %
\left[ L^{\ast },X\right] ,L\right] _{t}.
\end{equation}

Similarly $\left[ a_{-}\left( t\right) ,Y_{t}\right] =\frac{1}{2}\sqrt{n}%
\left[ L^{\ast },Y\right] _{t}$ and therefore we get the Wick re-ordering
\begin{eqnarray}
a_{-}\left( t\right) S^{\ast }\left[ X,L\right] _{t}=S^{\ast }\left[ X,L%
\right] _{t}a_{-}\left( t\right) +\frac{1}{2}\sqrt{n}\left[ L^{\ast },\left[
X,L\right] \right] _{t}.
\end{eqnarray}

This leads to the form of the quantum white noise equation with both $a_{+}$
and $a_{-}$ Wick ordered as

\begin{eqnarray}
\dot{X}_{t} &=&\sqrt{n+1}a_{+}\left( t\right) ^{\ast }S^{\ast }\left[ X,L%
\right] _{t}+\sqrt{n+1}\left[ L^{\ast },X\right] _{t}Sa_{+}\left( t\right)  
\notag \\
&&+\sqrt{n}a_{-}\left( t\right) ^{\ast }\left[ L^{\ast },X\right] _{t}S+%
\sqrt{n}S^{\ast }\left[ X,L\right] _{t}a_{-}\left( t\right)   \notag \\
&&+\mathcal{L}\left( X\right) _{t}+\frac{1}{2}n\left[ \left[ L^{\ast },X%
\right] ,L\right] _{t}+\frac{1}{2}n\left[ L^{\ast },\left[ X,L\right] \right]
_{t}.  \label{eq:approx_Heis}
\end{eqnarray}

At this stage we recognize (\ref{eq:approx_Heis}) as the equivalent form of
the Heisenberg quantum stochastic differential equation for thermal noise.

We also remark that the output process determined by systems in series is $%
B^{\mathrm{out}}\left( t\right) =U_{t}^{\ast }A_{+}\left( t\right) U_{t}$,
and from the quantum stochastic calculus we have 
\begin{equation}
dB^{\mathrm{out}}\left( t\right) =dA_{+}\left( t\right) +\left( L+S\sqrt{%
2\kappa k}c_{+}\right) _{t}dt.
\end{equation}
Using (\ref{eq:approx_eq}) we approximate this as 
\begin{equation}
dB^{\mathrm{out}}\left( t\right) \simeq dA_{+}\left( t\right) +L_{t}dt+S%
\frac{2\kappa }{\varepsilon ^{2}-\kappa ^{2}}\left[ \kappa dA_{+}\left(
t\right) +\varepsilon dA_{-}\left( t\right) ^{\ast }\right] \equiv SdB^{%
\mathrm{in}}\left( t\right) +L_{t}dt,
\end{equation}
that is, the thermal input $B^{\mathrm{in}}\left( t\right) =\sqrt{n+1}%
A_{+}\left( t\right) +\sqrt{n}A_{-}\left( t\right) ^{\ast }$ produces the
output $B^{\mathrm{out}}\left( t\right) $ according to the usual rules one
would expect of a quantum Markov component with the parameters $\mathbf{G}%
\sim \left( S,L,H\right) $.

Therefore the description of a component with the parameters $\mathbf{G}\sim
\left( S,L,H\right) $, at least in the case where $S$ is a static
beam-splitter matrix, with Gaussian input processes may be considered as the
same component cascaded with a degenerate parametric amplifier with vacuum
inputs in the singular coupling limit of the DPA.


\section{The General Series Product}

\subsection{\label{Sec:NS_SP}Without Scattering}

Let us now consider the situation where a Gaussian input $B_{\mathrm{in}}=B_{%
\mathrm{in}}^{\left( \mathscr{A}\right) }$ is driving a system with SLH
parameters $\left( I,L_{\mathscr{A}},H_{\mathscr{A}}\right) $ and that its
output $B_{\mathrm{out}}^{\left( \mathscr{A}\right) }$ acts as input $B_{%
\mathrm{in}}^{\left( \mathscr{B}\right) }$ to a second system $\left( I,L_{%
\mathscr{B}},H_{\mathscr{B}}\right) $. (We do not assume that any of the
various SLH operators commute!)

\bigskip

\textbf{(Components in Series: The no scattering case)} The Heisenberg QSDE
for the systems $\left( I,L_{\mathscr{A}},H_{\mathscr{A}}\right) $ and $%
\left( I,L_{\mathscr{B}},H_{\mathscr{B}}\right) $ given by 
\begin{equation}
dj_{t}(X)=\sum_{\mathscr{S}=\mathscr{A},\mathscr{B}}\bigg\{dB_{\mathrm{in}%
}^{\left( \mathscr{S}\right) \ast }\circ j_{t}(\left[ X,L_{\mathscr{S}}%
\right] )+j_{t}(\left[ L_{\mathscr{S}}^{\ast },X\right] )\circ dB_{\mathrm{in%
}}^{\left( \mathscr{S}\right) }+j_{t}(\mathcal{L}_{\mathscr{S}}X)\circ dt%
\bigg\},
\end{equation}
where 
\begin{equation}
\mathcal{L}_{\mathscr{S}}X=\frac{1}{2}L_{\mathscr{S}}^{\ast }\left[ X,L_{%
\mathscr{S}}\right] +\frac{1}{2}\left[ L_{\mathscr{S}}^{\ast },X\right] L_{%
\mathscr{S}}-i\left[ X,H_{\mathscr{S}}\right] .
\end{equation}
and we have the constraints $B_{\mathrm{in}}^{\left( \mathscr{A}\right) }=B_{%
\mathrm{in}}$ and $dB_{\mathrm{in}}^{\left( \mathscr{B}\right) }=dB_{\mathrm{%
in}}^{\left( \mathscr{A}\right) }+j_{t}\left( L_{\mathscr{A}}\right) dt$,
consistent with $B_{\mathrm{in}}$ driving system $\mathscr{A}$ which in turn
drives $\mathscr{B}$, corresponds to the dynamics given by the intrinsic
series product (\ref{eq:series_prod}).

\begin{proof}
We have to show consistency of the quantum stochastic Heisenberg evolution $%
j_{t}(\cdot )$. To this end we take the open loop equations and impose the
constraint $dB_{\mathrm{in}}^{\left( \mathscr{B}\right) }=dB_{\mathrm{in}%
}^{\left( \mathscr{A}\right) }+j_{t}\left( L_{\mathscr{A}}\right) dt$ giving 
\begin{eqnarray}
dj_{t}\left( X\right)  &=&dB_{\mathrm{in}}^{\ast }\circ j_{t}(\left[ X,L_{%
\mathscr{A}}\right] )+j_{t}(\left[ L_{\mathscr{A}}^{\ast },X\right] )\circ
dB_{\mathrm{in}}  \notag \\
&+&\left( dB_{\mathrm{in}}+j_{t}(L_{\mathscr{A}})\right) ^{\ast }\circ j_{t}(%
\left[ X,L_{\mathscr{B}}\right] )  \notag \\
&+&j_{t}(\left[ L_{\mathscr{B}}^{\ast },X\right] )\circ \left( dB_{\mathrm{in%
}}+j_{t}(L_{\mathscr{A}})dt\right)   \notag \\
&+&j_{t}(\mathcal{L}_{\mathscr{A}}X)\circ dt+j_{t}(\mathcal{L}_{\mathscr{B}%
}X)\circ dt,
\end{eqnarray}
which we may rearrange as 
\begin{eqnarray}
dj_{t}(X) &=&dB_{\mathrm{in}}^{\ast }\circ j_{t}(\left[ X,L_{\mathscr{A}}+L_{%
\mathscr{B}}\right] )+j_{t}(\left[ L_{\mathscr{A}}^{\ast }+L_{\mathscr{B}%
}^{\ast },X\right] )\circ dB_{\mathrm{in}} \nonumber \\
&&+j_{t}\bigg(\mathcal{L}_{\mathscr{A}}X+\mathcal{L}_{\mathscr{B}}X+L_{%
\mathscr{A}}^{\ast }\left[ X,L_{\mathscr{B}}\right] +\left[ L_{\mathscr{B}%
}^{\ast },X\right] L_{\mathscr{A}}\bigg)\circ dt.
\end{eqnarray}
However, the $dt$ term can be recast using the identity 
\begin{eqnarray}
&&\mathcal{L}_{\mathscr{A}}X+\mathcal{L}_{\mathscr{B}}X+L_{\mathscr{A}%
}^{\ast }\left[ X,L_{\mathscr{B}}\right] +\left[ L_{\mathscr{B}}^{\ast },X%
\right] L_{\mathscr{A}} \nonumber \\
&=&\frac{1}{2}\left( L_{\mathscr{A}}+L_{\mathscr{B}}\right) ^{\ast }\left[
X,L_{\mathscr{A}}+L_{\mathscr{B}}\right] +\frac{1}{2}\left[ L_{\mathscr{A}%
}^{\ast }+L_{\mathscr{B}}^{\ast },X\right] \left( L_{\mathscr{A}}+L_{%
\mathscr{B}}\right) \nonumber  \\
&&-i\left[ X,H_{\mathscr{A}}+H_{\mathscr{B}}+\frac{1}{2i}\left( L_{%
\mathscr{B}}^{\ast }L_{\mathscr{A}}-L_{\mathscr{A}}^{\ast }L_{\mathscr{B}%
}\right) \right] .
\end{eqnarray}
The resulting Heisenberg dynamics is therefore the same as for the model $%
(I,L,H)$ with $L=L_{\mathscr{A}}+L_{\mathscr{B}}$, and $H=H_{\mathscr{A}}+H_{%
\mathscr{B}}+\mathrm{Im}\{L_{\mathscr{B}}^{\ast }L_{\mathscr{A}}\}$. This
is, of course, the form predicted by the series product in the Fock case (%
\ref{eq:series_prod}).
\end{proof}

\subsection{Including Scattering}

As mentioned above, it is not possible to construct a well defined
scattering processes $\Lambda _{jk}$ in the non-Fock theory. Nevertheless,
the effects of static beam-splitter scattering $S$ may be included in a
straightforward manner without directly considering unitary QSDE models
involving the scattering processes. A clue on how to proceed is given by
our earlier observation that if the scattering matrix $S$ entries commute with systems operators - physically, a static beam-splitter -
the scattering processes disappears.

\bigskip

In the Fock representation, we could always take the input field $A_{\mathrm{%
in}}$ and apply a unitary rotation $A=SA_{\mathrm{in}}$ before passing it
though as drive for component. As we have seen, this
will require a compensating rotation of the coupling operators, but no
change to the Lindbladian. There is also a rotation of the output, however,
anticipating this we make the following definition.

\begin{definition}{Definition}
Let $\mathbf{G}$ and $\mathbf{\tilde{G}}$ be SLH model parameters which, for
given input noise $A_{\mathrm{in}}=\tilde{A}_{\mathrm{in}}$ lead to output
noises $A_{\mathrm{out}}$ and $\tilde{A}_{\mathrm{out}}$ respectively. We say
that the models' input-output relations are \textbf{related by a static
beam-splitter matrix} $S$ if we have 
\begin{eqnarray}
A_{\mathrm{out}}=S\,\tilde{A}_{\mathrm{out}}.
\end{eqnarray}
\end{definition}

The following result shows that for the Fock representation, if the
scattering is just a static beam-splitter, then we can produce a related
model which avoids the use of the scattering processes.

\bigskip

\begin{theorem}{Theorem}
\label{Thm:S} Let $S$ be a static beam-splitter matrix and set $\mathbf{G}%
\sim \left( S,L,H\right) $ and $\mathbf{\tilde{G}}\sim \left( I,S^{\ast
}L,H\right) $. Then the model parameters $\mathbf{G}$ and $\mathbf{\tilde{G}}
$ generate the same Heisenberg dynamics. Moreover, their input-output
relations are related by the static beam-splitter matrix $S$.
\end{theorem}

\begin{proof}
The Heisenberg dynamics generated by $\mathbf{G}$ is (the scattering terms
vanish for a static beam-splitter) 
\begin{equation}
dj_{t}^{\mathbf{G}}(X)=\sum_{j}j_{t}(\mathcal{L}_{j0}^{\mathbf{G}%
}X)\,dA_{j}^{\ast }+\sum_{k}j_{t}(\mathcal{L}_{0k}^{\mathbf{G}%
}X)\,dA_{k}+j_{t}(\mathcal{L}^{\mathbf{G}}X)dt
\end{equation}
where 
\begin{equation}
\mathcal{L}_{j0}^{\mathbf{G}}X=S_{lj}^{\ast }\left[ X,L_{l}\right] ,\quad 
\mathcal{L}_{0k}^{\mathbf{G}}X=\left[ L_{l}^{\ast },X\right] S_{lk}
\end{equation}
and the Lindblad generator is $\mathcal{L}^{\mathbf{G}}X=\frac{1}{2}%
L_{k}^{\ast }\left[ X,L_{k}\right] +\frac{1}{2}\left[ L_{k}^{\ast },X\right]
L_{k}-i\left[ X,H\right] $. The Heisenberg dynamics for $\mathbf{\tilde{G}}$
similarly has no scattering terms in its QSDE, and we see that 
\begin{equation}
\mathcal{L}_{j0}^{\mathbf{G}}X=[X,S_{lj}^{\ast }L_{l}]\equiv \mathcal{L}%
_{j0}^{\mathbf{\tilde{G}}}X,\quad \mathcal{L}_{0k}^{\mathbf{G}%
}X=[L_{l}^{\ast }S_{lk},X]\equiv \mathcal{L}_{0k}^{\mathbf{\tilde{G}}}X.
\end{equation}
From the unitarity and scalar nature of $S$ we have that 
\begin{eqnarray}
\mathcal{L}^{\mathbf{\tilde{G}}}X &=&\frac{1}{2}L_{k}^{\ast }S_{kl}\left[
X,S_{jl}^{\ast }L_{j}\right] +\frac{1}{2}\left[ L_{k}^{\ast }S_{kl},X\right]
S_{jl}^{\ast }L_{j}-i\left[ X,H\right]  \nonumber\\
&=&\frac{1}{2}L_{k}^{\ast }\left[ X,L_{k}\right] +\frac{1}{2}\left[
L_{k}^{\ast },X\right] L_{k}-i\left[ X,H\right]  \nonumber \\
&\equiv &\mathcal{L}^{\mathbf{G}}X.
\end{eqnarray}
Therefore the QSDEs corresponding to the Heisenberg dynamics for $\mathbf{G}$
and $\mathbf{\tilde{G}}$ are identical.

The input-output relations for $\mathbf{G}$ are 
\begin{equation}
dA_{\mathrm{out},j}\left( t\right) =S_{jk}\,dA_{\mathrm{in},k}+j_{t}\left(
L_{j}\right) \,dt
\end{equation}
while for $\mathbf{\tilde{G}}$ we have 
\begin{equation}
dB_{\mathrm{out},j}\left( t\right) =dB_{\mathrm{in},j}+S_{jk}\,j_{t}\left(
L_{k}\right) \,dt.
\end{equation}
If we require the inputs to be the same ($A_{\mathrm{in}}=B_{\mathrm{in}}$)
then we have $A_{\mathrm{out}}=S\,B_{\mathrm{out}}$.
\end{proof}

\bigskip

Our strategy for introducing static beam-splitter scattering into the
situation where we have non-Fock noise input fields is to say that the
initial input $A_{\mathrm{in}}$ be replaced by the rotated input $SA_{%
\mathrm{in}}$, and exploit the fact that the Heisenberg dynamics no longer
involves the scattering processes $\Lambda _{jk}$ explicitly.

\begin{lemma}{Lemma}
\textbf{(The Universal Heisenberg QSDE Description)} The Heisenberg dynamics
for a general $\left( S,L,H\right) $ model with a static beam-splitter
matrix $S$ are given by the QSDE 
\begin{eqnarray}
dj_{t}(X)&=&dA_{\mathrm{in}}^{\ast }\circ S^{\ast }j_{t}(\left[ X,L\right]
)+j_{t}(\left[ L^{\ast },X\right] )S\circ dA_{\mathrm{in}}  +j_{t}(\mathcal{L}X)\circ dt
\end{eqnarray}
for all mean-zero Gaussian input fields $A_{\mathrm{in}}$.
\end{lemma}

\bigskip

This is of course just the equation (\ref{eq:approx_Heis}) written in the Wick-Stratonovich form so as to be
representation free!

\bigskip

Now let us try and repeat or analysis from Section \ref{Sec:NS_SP}. Let us
now consider the situation where a Gaussian input $A_{\mathrm{in}}=A_{%
\mathrm{in}}^{\left( 1\right) }$ is driving a system with SLH parameters $%
\left( S_{\mathscr{A}},L_{\mathscr{A}},H_{\mathscr{A}}\right) $ and that its
output $A_{\mathrm{out}}^{\left( 1\right) } $ acts as input for a second
system $\left( S_{\mathscr{B}},L_{\mathscr{B}},H_{\mathscr{B}}\right) $.

\begin{lemma}{Lemma}
\label{prop:fig} \textbf{(Components in series: With a static beam-splitter
scattering)} The Heisenberg QSDE for a pair of systems $\left( S_{\mathscr{A}%
},L_{\mathscr{A}},H_{\mathscr{A}}\right) $ and $\left( S_{\mathscr{B}},L_{%
\mathscr{B}},H_{\mathscr{B}}\right) $ in series is 
\begin{eqnarray}
dj_{t}(X)=\sum_{\mathscr{S}=\mathscr{A},\mathscr{B}}\bigg\{dA_{\mathrm{in}%
}^{\left( \mathscr{S}\right) \ast }\circ j_{t}(\left[ X,L_{\mathscr{S}}%
\right] )  
+j_{t}(\left[ L_{\mathscr{S}}^{\ast },X\right] )\circ dA_{\mathrm{in}%
}^{\left( \mathscr{S}\right) }+j_{t}(\mathcal{L}_{\mathscr{S}}X)\circ dt%
\bigg\} ,
\end{eqnarray}
where $A_{\mathrm{in}}^{\left( \mathscr{A}\right) }=S_{\mathscr{A}}A_{\mathrm{in}%
}$ and $A_{\mathrm{in}}^{\left( \mathscr{B}\right) }=S_{\mathscr{B}}A_{\mathrm{%
out}}^{\left( \mathscr{A}\right) }$ where $dA_{\mathrm{out}}^{\left( %
\mathscr{A}\right) }=S_{\mathscr{A}}dA_{\mathrm{in}}^{\left( \mathscr{A}%
\right) }+j_{t}\left( L_{\mathscr{A}}\right) dt$, and the Lindbladians $\mathcal{L}_{\mathscr{S}}$ are as before.
\end{lemma}

\begin{proof}
Substituting the processes into the QSDEs yields 
\begin{eqnarray}
dj_{t}(X) &=&\left( S_{\mathscr{A}}dA_{\mathrm{in}}\right) ^{\ast }\circ
j_{t}(\left[ X,L_{\mathscr{A}}\right] ) +j_{t}(\left[ L_{\mathscr{A}}^{\ast
},X\right] )\circ S_{\mathscr{A}}dA_{\mathrm{in}}  \nonumber\\
&&+\left( S_{\mathscr{B}}S_{\mathscr{A}}dA_{\mathrm{in}}+S_{\mathscr{B}}L_{%
\mathscr{A}}dt\right) ^{\ast }\circ j_{t}(\left[ X,L_{\mathscr{B}}\right] ) \nonumber
\\
&& +j_{t}(\left[ L_{\mathscr{B}}^{\ast },X\right] )\circ \left( S_{%
\mathscr{B}}S_{\mathscr{A}}dA_{\mathrm{in}}+S_{\mathscr{B}}L_{\mathscr{A}%
}dt\right) \nonumber \\
&&+j_{t}(\mathcal{L}_{\mathscr{A}}X)\circ dt+j_{t}(\mathcal{L}_{\mathscr{B}%
}X)\circ dt, \nonumber \\
&=&\left( dA_{\mathrm{in}}\right) ^{\ast }\circ j_{t}(\left[ X,S_{\mathscr{A}%
}^{\ast }L_{\mathscr{A}}+S_{\mathscr{A}}^{\ast }S_{\mathscr{B}}^{\ast }L_{%
\mathscr{B}}\right] ) +j_{t}(\left[ L_{\mathscr{A}}^{\ast }S_{\mathscr{A}%
}+L_{\mathscr{B}}^{\ast }S_{\mathscr{B}}S_{\mathscr{A}},X\right] )\circ dA_{%
\mathrm{in}} \nonumber \\
&&+j_{t}(\mathcal{L}_{\mathscr{A}}X+\mathcal{L}_{\mathscr{B}}X +L_{%
\mathscr{A}}^{\ast }S_{\mathscr{B}}^{\ast }\left[ X,L_{\mathscr{B}}\right] +%
\left[ L_{\mathscr{B}}^{\ast },X\right] S_{\mathscr{B}}L_{\mathscr{A}})\circ
dt.
\end{eqnarray}

A similar calculation to before shows that 
\begin{eqnarray}
&&\mathcal{L}_{\mathscr{A}}X+\mathcal{L}_{\mathscr{B}}X+L_{\mathscr{A}%
}^{\ast }S_{\mathscr{B}}^{\ast }\left[ X,L_{\mathscr{B}}\right] +\left[ L_{%
\mathscr{B}}^{\ast },X\right] S_{\mathscr{B}}L_{\mathscr{A}} \nonumber \\
&=&\frac{1}{2}\left( S_{\mathscr{B}}L_{\mathscr{A}}+L_{\mathscr{B}}\right)
^{\ast }\left[ X,S_{\mathscr{B}}L_{\mathscr{A}}+L_{\mathscr{B}}\right] +%
\frac{1}{2}\left[ L_{\mathscr{A}}^{\ast }S_{\mathscr{B}}^{\ast }+L_{%
\mathscr{B}}^{\ast },X\right] \left( S_{\mathscr{B}}L_{\mathscr{A}}+L_{%
\mathscr{B}}\right) \nonumber  \\
&&-[iX,H_{\mathscr{A}}+H_{\mathscr{B}}+\frac{1}{2i}\left( L_{\mathscr{B}%
}^{\ast }S_{\mathscr{B}}L_{\mathscr{A}}-L_{\mathscr{A}}^{\ast }S_{\mathscr{B}%
}^{\ast }L_{\mathscr{B}}\right) ].
\end{eqnarray}
The resulting Heisenberg dynamics is therefore same as for the model $%
\mathbf{\tilde{G}}\sim (I,\tilde{L},H)$ with coupling operators $\tilde{L}%
=S_{\mathscr{A}}^{\ast }L_{\mathscr{A}}+S_{\mathscr{A}}^{\ast }S_{\mathscr{B}%
}^{\ast }L_{\mathscr{B}}\equiv S_{\mathscr{A}}^{\ast }S_{\mathscr{B}}^{\ast
}\left( S_{\mathscr{B}}L_{\mathscr{A}}+L_{\mathscr{B}}\right) $, and
Hamiltonian $H=H_{\mathscr{A}}+H_{\mathscr{B}}+\mathrm{Im}\{L_{\mathscr{B}%
}^{\ast }S_{\mathscr{B}}L_{\mathscr{A}}\}$.

The output is then $B_{\mathrm{out}}$ where 
\begin{equation}
dB_{\mathrm{out}}\left( t\right) =dA_{\mathrm{in}}\left( t\right)
+j_{t}\left( S_{\mathscr{A}}^{\ast }L_{\mathscr{A}}+S_{\mathscr{A}}^{\ast
}S_{\mathscr{B}}^{\ast }L_{\mathscr{B}}\right) dt.
\end{equation}
The correct output for this should however be $A_{\mathrm{out}}=S_{%
\mathscr{B}}S_{\mathscr{A}}B_{\mathrm{out}}$ so that

\begin{eqnarray}
dA_{\mathrm{out}}\left( t\right) =S_{\mathscr{B}}S_{\mathscr{A}}dA_{\mathrm{%
in}}\left( t\right) +j_{t}\left( S_{\mathscr{B}}L_{\mathscr{A}}+L_{%
\mathscr{B}}\right) dt
\end{eqnarray}
and we have the desired matrix $S_{\mathscr{B}}S_{\mathscr{A}}$ multiply the
inputs corresponding to scattering first by matrix $S_{\mathscr{A}}$ and
then by $S_{\mathscr{B}}$. The model $\mathbf{G}$ obtained from postulate Ia
is then the one related to $\mathbf{\tilde{G}}$ by the static beam-splitter
matrix $S_{\mathscr{B}}S_{\mathscr{A}}$, that is (from Theorem \ref{Thm:S}
with $S=S_{\mathscr{B}}S_{\mathscr{A}}$ and $\tilde{L}=S^\ast L$%
\begin{eqnarray}
\mathbf{G} & \sim &  \left( S,L,H\right) =\left( S_{\mathscr{B}}S_{\mathscr{A}%
},S_{\mathscr{B}}S_{\mathscr{A}}\tilde{L},H\right)  \notag \\
&=&
\left( S_{\mathscr{B}}S_{\mathscr{A}},S_{\mathscr{B}}L_{\mathscr{A}}+L_{%
\mathscr{B}},H_{\mathscr{A}}+H_{\mathscr{B}}+\mathrm{Im}\{L_{\mathscr{B}%
}^{\ast }S_{\mathscr{B}}L_{\mathscr{A}}\}\right) ,  \notag
\end{eqnarray}
and again we have the same form as the series product in the Fock case (\ref
{eq:series_prod}).
\end{proof}


\section{Conclusions}

We have shown that there is a consistent theory for quantum input-output
models in series when the driving input processes are in general Gaussian
states with a flat power spectrum. This emerges fairly explicitly at the
level of the singular input processes $b_k (t)$ themselves, but to have a
working theory we need to make the connection to the Hudson-Parthasarathy
quantum stochastic calculus. This involves quantum stochastic differential
equations on the Fock spaces used to represent the noise (which are a
mathematical convenience and not physical objects) with the result that the
associated dynamical equations appear to depend on the choice of Gaussian
state of the noise. In reality this is a mathematical artifact and we show
that even here there is a way of expressing the quantum stochastic
differential equations (the Wick-Stratonovich form introduced in this paper)
which removes these terms. In effect, it is the Wick-Stratonovich form that
translates in the physically relevant dynamical equations written in terms
of the quantum input processes $b_k (t)$.

The connection rules are then shown to be genuinely independent of the
choice of state. We were also able to include the effects of a static
beam-splitter component. At first sight this would seem problematic as the
scattering terms $\Lambda _{jk}(t)$ are not well-defined for non-vacuum
states, however, it is possible to ignore them from the model: in fact we
need to work at the level of the Heisenberg flow and the input-output
relations, neither of which involve the scattering terms. The result is that
we may account for static scattering and we find that the series product of 
\cite{GJ-Series} again gives the correct rule. In this way we extend the
series product to deal with quantum feedback networks driven by general
Gaussian input processes.

We have restricted our analysis to Bose systems, however, there is an Araki-Woods
type double Fock space representation for Fermi fields with quasi-free states as well,
and is applicable to Fermi stochastic processes \cite{HP_Fermi}, \cite{BSW}.
The network rules for Fermi stochastic processes can be similarly derived
and one would naturally expect these to again be state-independent.

\end{document}